\documentclass[aps,pra,10pt,notitlepage,a4paper,twocolumn,superscriptaddress]{revtex4-2}
\usepackage{amsfonts,amssymb,amsmath}       
\usepackage{lmodern}
\usepackage[]{graphics,graphicx}            
\usepackage{amsthm}
\usepackage{bbold,enumitem,mathtools}
\usepackage[breakable]{tcolorbox}
\usepackage{physics}
\usepackage{tikz}
\usetikzlibrary{decorations.pathmorphing,patterns,decorations.markings,matrix,quantikz2}
\usepackage{adjustbox}
\usepackage{xstring}
\usepackage{ifthen}
\usepackage{listings}
\usepackage{booktabs}
\usepackage{todonotes}
    \usepackage{hyperref}
\usepackage[capitalize]{cleveref}

\usepackage[english]{babel}

\makeatletter
\def\bbl@set@language#1{%
  \edef\languagename{%
    \ifnum\escapechar=\expandafter`\string#1\@empty
    \else\string#1\@empty\fi}%
  \@ifundefined{babel@language@alias@\languagename}{}{%
    \edef\languagename{\@nameuse{babel@language@alias@\languagename}}%
  }%
  \select@language{\languagename}%
  \expandafter\ifx\csname date\languagename\endcsname\relax\else
    \if@filesw
      \protected@write\@auxout{}{\string\select@language{\languagename}}%
      \bbl@for\bbl@tempa\BabelContentsFiles{%
        \addtocontents{\bbl@tempa}{\xstring\select@language{\languagename}}}%
      \bbl@usehooks{write}{}%
    \fi
  \fi}
\newcommand{\DeclareLanguageAlias}[2]{%
  \global\@namedef{babel@language@alias@#1}{#2}%
}
\makeatother

\DeclareLanguageAlias{en}{english}

\newtheorem{theorem}{Theorem}

\newtheorem{lemma}{Lemma}

\newtcolorbox{example}[1][]{width=0.48\textwidth,boxrule=0mm,leftrule=1mm,colframe=black!75,sharp corners,before=\par\smallskip\centering,after=\par,title=#1,breakable,left=2pt,right=2pt}

\newenvironment{proof*}[1][\proofname]{%
  
  \begin{proof}[#1]}{\end{proof}}

\renewcommand{\epsilon}{\varepsilon}

\newcommand{\ignore}[1]{}

\pgfset{
  foreach/parallel foreach/.style args={#1in#2via#3}{evaluate=#3 as #1 using {{#2}[#3-1]}},
}

\newcounter{nodenumber}
\newcommand{\drawchain}[2][0.4\textwidth]{%
\def\firstlist{#2}
\begin{center}
\begin{adjustbox}{width=#1}
\begin{tikzpicture}
  \setcounter{nodenumber}{0}
  \foreach \x [count=\c]  in \firstlist
  {%
    \draw [thick] (2*\c,0) -- node[above,pos=0.5] {\x} (2*\c+2,0);
    \stepcounter{nodenumber}
    \node[circle,style={fill=black,minimum width=0.8cm,text=white}] at (2*\c,0) {};
  }
  \node[circle,style={fill=black,minimum width=0.8cm,text=white}] at (2*\arabic{nodenumber}+2,0) {};

\end{tikzpicture}
\end{adjustbox}
\end{center}
}
\newcounter{arraycard}

\forceredefine{}

\begin{document}

\title{Optimising Perfect Quantum State Transfer for Timing Insensitivity}
\date{\today}
\author{Alastair \surname{Kay}}\email{alastair.kay@rhul.ac.uk}
\affiliation{Royal Holloway University of London, Egham, Surrey, TW20 0EX, UK}

\author{Sooyeong \surname{Kim}}\email{sooyeong@uoguelph.ca}
\affiliation{Department of Mathematics and Statistics, University of Guelph, Guelph, ON, Canada N1G 2W1}

\author{Christino Tamon}
\email{ctamon@clarkson.edu}
\affiliation{Department of Computer Science, Clarkson University, Potsdam, New York, USA 13699-5815}

\begin{abstract}
When studying the perfect transfer of a quantum state from one site to another, it is typically assumed that one can receive the arriving state at a specific instant in time, with perfect accuracy. Here, we study how sensitive perfect state transfer is to that timing. We design engineered spin chains which reduce their sensitivity, proving that this construction is asymptotically optimal. The same construction is applied to the task of creating superpositions, also known as fractional revival.
\end{abstract}
\maketitle

Studies of perfect quantum state transfer (PST) \cite{bose2003,christandl2004,christandl2005} were originally motivated by a desire to simplify the act of moving quantum states around within a small quantum device, helping distant qubits to interact when only local interactions are available. Compared to the gate model's sequence of swaps, state transfer can be significantly faster and less sensitive to timing errors. However, little has been done to {\em optimise} this insensitivity; it was asserted in \cite{kay2010a,kay2006b} that the original perfect transfer chain \cite{christandl2004} was essentially optimal. This analysis was too na\"ive, as we will show here. Since then, Kirkland \cite{kirkland2015} has provided tools to characterise the timing sensitivity, but does not address the optimisation problem. Other works \cite{lippnershi2024,lippnershi2025} have observed that it is possible to get very broad arrival peaks, although the method inherently introduces a massive cost that entirely precludes such schemes from practical consideration (the required time scales exponentially in the distance of transfer).

In this paper, we design a family of perfect transfer $N$-qubit chains which are dependent upon an additional parameter $\gamma$. We prove that in the large $\gamma$ limit, these chains are asymptotically optimal in the sense that they optimise the possible trade-off between state transfer time and arrival width, for any length of spin chain. We also show how this construction can be modified to allow the creation of a state superposed between opposite ends of the chain, a task known as fractional revival \cite{chen2007,dai2010,kay2010a,banchi2015}, also with broad arrival characteristics. In comparison to the original perfect transfer chains, which have tightly peaked arrival statistics of the form $\sin^{2(N-1)}(t)$, our new designs achieve, in the limit, $\sin^6(t)$ (or $\sin^8(t)$) for any even (or odd) chain length $N$ of at least 4. Shorter perfect transfer chains are uniquely defined, so there is no optimisation to be performed. We leave largely open the question of chains with high fidelity transfer, faster than the fastest perfect transfer, with broad arrival width, except to numerically check the performance of some existing schemes. All detailed calculations are supported by examples, further details of which may be found in \cite{kay2025d}.

\section{Figure of Merit}

In the task of state transfer on a $N$-qubit spin chain, a Hamiltonian $H$ evolves an unknown quantum state, $\ket{\psi}$, prepared on a single qubit $A$, aiming for that state to arrive at another site $B$ after a fixed period of time $t_0$:
$$
e^{-iHt_0}\ket{\psi}_A\ket{0}^{\otimes N-1}=\ket{0}^{\otimes N-1}\ket{\psi}_B.
$$
For an excitation preserving Hamiltonian,
$$
\left[H,\sum_{n=1}^NZ_n\right]=0,
$$
this problem is reduced to excitation transfer within an $N$-dimensional subspace on which the Hamiltonian is $H_0$:
$$
e^{-iH_0t_0}\ket{1}=\ket{N}.
$$
There is the possibility to include a phase in the arrival, but we will not unduly worry about that here.

We restrict to a nearest-neighbour interaction such that $H_0$ is a tridiagonal matrix, and we consider the field-free case of 0 on the diagonal.
To achieve perfect transfer under these assumptions, the necessary and sufficient conditions are known \cite{kay2010a}: the chain must be symmetric, i.e.\ $SHS=H$, where the symmetry operator is
$$
S=\sum_{n=1}^{N} \ket{N+1-n}\bra{n},
$$
and the spectrum, up to a scale factor $\frac{\pi}{t_0}$, has odd integer gaps between consecutive eigenvalues. Any such spectrum yields a perfect transfer chain -- one just uses the Lanczos algorithm to rebuild the symmetric chain from the specified eigenvalues \cite{karbach2005,kay2010a,gladwell1986}.

In order to allow for imperfect transfer, whether this is due to an improperly designed $H_0$, error in manufacture, noise, or a timing imperfection, we need to quantify success. At a single moment in time, $t$, the state arriving at site $B$ may be described by a density matrix $\rho(t)$, and the corresponding figure of merit is the fidelity \cite{bose2003,christandl2005,kay2010a},
$$
F(t)=\bra{\psi}\rho(t)\ket{\psi}=\frac13+\frac{(1+\sqrt{F_e(t)})^2}{6}.
$$
where $F_e(t)$ is the fidelity of excitation transfer,
$$
F_e(t)=|\bra{N}e^{-iH_0t}\ket{1}|^2.
$$

We are interested in how successful we are if timing goes a little bit wrong. To that end, we characterise the receiver's operations by a function $p(t)$, which specifies the probability that the state is received at time $t$. We imagine this $p(t)$ to be a fixed function of the device and, whatever time one is aiming to receive the state at, $t_0$, we receive the state at time $t$ with probability $p(t-t_0)$. The natural figure of merit is the expected fidelity,
\begin{align*}
\bar F&=\int_0^{\infty}F(t)p(t-t_0)dt\\&=\frac12+\frac16\int_0^{\infty}p(t-t_0)\left(2\sqrt{F_e(t)}+F_e(t)\right)dt.
\end{align*}
For simplicity, we consider the closely related quantity
$$
\tilde F_e=\int_0^{\infty}p(t-t_0)|\bra{N}e^{-iH_0t}\ket{1}|dt.
$$

While the function $p(t)$ might vary from device to device, we generically expect it to be strongly peaked at $t=0$. Hence, it makes sense to perform a small time expansion about this point:
\begin{multline*}
\tilde F_e=\int_0^{\infty}p(\delta t)\left|\bra{N}(e^{-iH_0t_0}-iH_0\delta te^{-iH_0t_0}\right.\\
\left.-H_0^2\frac{\delta t^2}{2}e^{-iH_0t_0}+\order{\delta t^3})\ket{1}\right|dt.
\end{multline*}
If we focus on perfect state transfer chains with a transfer time of $t_0$, then $\bra{N}e^{-iH_0t_0}=\bra{1}$, and thus,
$$
\tilde F_e\approx\int_0^{\infty}p(\delta t)\left|1-\frac{\delta t^2}{2}\bra{1}H_0^2\ket{1}\right|dt.
$$
For a field-free Hamiltonian, $\bra{1}H_0^2\ket{1}=J_{1}^2$, the first (and last) coupling strength. In other words, to leading order, we should aim to minimise the first coupling strength. The symmetry between first and last coupling strengths means that we simultaneously optimise the departure and arrival characteristics.

In order to ensure a fair comparison between different spin chains, we must impose some further constraints. Our aim is to find the best perfect state transfer chain for a given transfer distance (i.e.\ length of chain). We must impose a common maximum coupling strength $J_{\max}$. This already fixes that there is a single perfect transfer chain with the minimum possible transfer time $t_{\min}$ \cite{yung2006,kay2016b}, corresponding to the original perfect transfer chain \cite{christandl2004}. Any other solution must be slower, and we are thus interested in the trade-off between transfer speed and insensitivity to timing errors.

The original analysis due to Kay \cite{kay2010a,kay2006b} started from this approximation and decomposed the behaviour in terms of the spectrum $\{\lambda_n\}$ of $H_0$, and the corresponding weights of the eigenvectors on the first/last site
$$
a_n=|\braket{1}{\lambda_n}|^2=|\braket{N}{\lambda_n}|^2.
$$
Given that
$$
J_1^2=\sum_{n=1}^{N} \lambda_n^2a_n,
$$
if it were possible to independently optimise $\{\lambda_n\}$ and $\{a_n\}$, then you should centre the $\lambda_n$ on 0 and give them as small a spread as possible, subject to the perfect state transfer condition that each eigenvalue is separated by at least $\frac{\pi}{t_0}$. This is achieved by the original chain, where all the gaps are exactly $\frac{\pi}{t_0}$. In fact,
\begin{equation}\label{eq:evecels}
a_n\propto\frac{1}{|q'(\lambda_n)|},
\end{equation}
where the characteristic polynomial of $H_0$ is
$$
q(x)=\prod_{n=1}^{N} (x-\lambda_n).
$$
The fallacy of the argument is thus revealed; the $a_n$ and the $\lambda_n$ are not independent. This then begs the question, can one do any better? Perhaps by allowing greater separation of the $\lambda_n$, the corresponding $a_n$ can be made smaller (subject to the constraint $\sum_na_n=1$). This is exactly what we will achieve, once we have quantified what is the best that one could hope for.

\section{Limiting Behaviour}

One key point of comparison is the original family perfect state transfer solutions \cite{christandl2004,christandl2005,albanese2004}, known as the Krawtchouk chains, which have coupling strengths
\begin{equation}\label{eq:kraw_couple}
J_n=J\sqrt{n(N-n)}.
\end{equation}
$J$ is a scale factor.
The excitation transfer fidelity is
\begin{equation}\label{eq:Kraw_profile}
F_e=\sin^{2(N-1)}(Jt).
\end{equation}
Given that these are the unique solution with the fastest possible perfect transfer \cite{yung2006,kay2016b}, they must be a limiting case. Indeed, for $N=2,3$, they are the unique solutions to perfect state transfer, and are certainly optimal in terms of the sensitivity to timing errors for those particular lengths. In the following subsection, we will prove that the $N=4,5$ cases are also optimal (i.e.\ have the smallest possible $J_1$ for a fixed transfer time), while the primary purpose of the rest of the paper is to construct chains improve upon these ones at larger $N$.

If a chain has perfect transfer at time $t_0$, then we know that by $t_0$, the initial state $\ket{1}$ has evolved into an orthogonal state $\ket{N}$.
The Mandelstam–Tamm theorem \cite{mandelstam1991} states that for an initial state $\ket{\psi}$ to evolve into an orthogonal one, it must take a time at least
$$
\frac{\pi}{2\sqrt{\bra{\psi}H_0^2\ket{\psi}-\bra{\psi}H_0\ket{\psi}^2}}
$$
which, in the present case, is just $\frac{\pi}{2J_1}$.
The Krawtchouk perfect transfer chain, taking $J_{\max}=1$, 
has $J_1=\order{1/\sqrt{N}}$,
which already gives us a very small value of $\bra{N}H_0^2\ket{N}$, but is far from saturating this limit, taking time $\order{N}$ rather than the limit of $\order{\sqrt{N}}$.

The Mandelstam-Tamm theorem was generalised in \cite{anandan1990}, giving us a statement
$$
F_e\leq\sin^2(J_1t).
$$
This shows the optimal arrival function as a function of time, while also letting us handle the case where there is not perfect transfer during the evolution.

\subsection{Improved Bound}

The Mandelstam-Tamm theorem is quite effective, but applies in a very general setting. By imposing perfect transfer, a stronger limit can be found.

\begin{theorem}\label{thm:main}
    For a chain of odd (or even) length $N\geq 4$ with perfect excitation transfer in time $t_0$,
    $$
    J_1^2\geq \frac{\pi\alpha}{2t_0}
    $$
where $\alpha=2$ (or $\alpha=\sqrt{3}$).
\end{theorem}

Before we can prove this, we need a helpful result:
\begin{lemma}\label{lem:family}
For $k \ge 1$,
let $H^{(k)}$ be the family of field-free tridiagonal matrices of size $N_k$ (where $N_{k+1}=N_k+2$ with $N_1=4$ or $5$) that exhibit perfect state transfer in time $t_0$ and have the smallest possible first coupling strength $J_1^{(k)}$. Then, it must be that
$
J_1^{(k+1)}>J_1^{(k)}.
$
\end{lemma}
\begin{proof}
We can start with $H^{(k+1)}$ and use it to construct a new $\tilde H^{(k)}$ which has perfect transfer and a value of $\tilde J_1<J_1^{(k+1)}$. By definition, it must be that $J_1^{(k)}\leq \tilde{J}_1$.

To construct $\tilde H^{(k)}$, take the ordered eigenvalues $\lambda_n^{(k+1)}$ of $H^{(k+1)}$, remove $\pm\lambda_1^{(k+1)}$ from the spectrum, and find the corresponding symmetric chain.

Such a chain certainly has PST in the time $t_0$ because we have not altered the spacings of any of the remaining eigenvalues. What is its value of $J_1^2$? We can write the first elements of the eigenvectors as
$$
a_n=a_{n+1}^{(k+1)}({\lambda^{(k+1)}_1}^2-\lambda_n^2)/\Gamma
$$
using \cref{eq:evecels}, where $\Gamma$ allows us to renormalise such that:
\begin{align*}
1&=\sum_{n=1}^{2k}a_n\\
&=\frac{1}{\Gamma}\sum_{n=2}^{2k+1} a_{n}^{(k+1)} ({\lambda^{(k+1)}_1}^2-{\lambda_n^{(k+1)}}^2) \\
&=\frac{1}{\Gamma} ({\lambda^{(k+1)}_1}^2-{J_1^{(k+1)}}^2).
\end{align*}
Finally, we can work out the new coupling strength
\begin{align*}
\tilde J_1^2
&=\frac{1}{\Gamma}\sum_{n=2}^{2k+1} a_{n}^{(k+1)}({\lambda^{(k+1)}_1}^2-{\lambda_n^{(k+1)}}^2) {\lambda_n^{(k+1)}}^2\\
&={J_1^{(k+1)}}^2-\frac{{J_1^{(k+1)}}^2{J_2^{(k+1)}}^2}{\Gamma}
< {J_1^{(k+1)}}^2.
\end{align*}
It follows that ${J_1^{(k)}}\le \tilde{J}_1 < {J_1^{(k+1)}}$.
\end{proof}

 \begin{proof}[Proof of Theorem \ref{thm:main}]
We now know that for a fixed perfect transfer time $t_0$, and a given chain length $N_k\geq 4$, $J^{(1)}_1\leq J^{(k)}_1$. Our aim is to determine $J_1^{(1)}$.

In the even length case, we could reduce all the way down to length 2, recovering the Mandelstam-Tamm limit. For a tighter bound, consider a length 4 chain with couplings $\{J_1,J_2,J_1\}$. This has eigenvalues
$$
\frac12\left(\pm J_2\pm\sqrt{4J_1^2+J_2^2}\right).
$$
There are therefore gaps in this system of size $J_2$ and $\sqrt{4J_1^2+J_2^2}-J_2$, both of which must be odd multiples of $\frac{\pi}{t_0}$. Thus, $J_2=(2m+1)\pi/t_0$ and $\sqrt{4J_1^2+J_2^2}=2n\pi/t_0$ for non-negative integers $n,m$. Rearranging for $J_1$,
$$
J_1^2=\frac{\pi^2}{4t_0^2}(4n^2-(2m+1)^2),
$$
we pick $n=1,m=0$ to minimise $J_1^2$, conveying that
$$
J_1^{(m)}\geq J_1\geq\frac{\pi\sqrt{3}}{2t_0}.
$$

In the case of odd $N$, we take a chain of length 5 with couplings $\{J_1,J_2,J_2,J_1\}$, which has eigenvalues
$$
0,\pm J_1,\pm\sqrt{J_1^2+2J_2^2}.
$$
Since each consecutive gap is an odd multiple of $\frac{\pi}{t_0}$,
$$
J_1^{(k)}\geq J_1=(2m+1)\frac{\pi}{t_0}\geq\frac{\pi}{t_0}.
$$

Incidentally, this proves that the perfect state transfer chains of length 4,5 with the optimal value of $\bra{1}H_0^2\ket{1}$ are the Krawtchouk chains \cite{christandl2004}.
\end{proof}

\section{The T-Rex Construction}\label{sec:clearout}

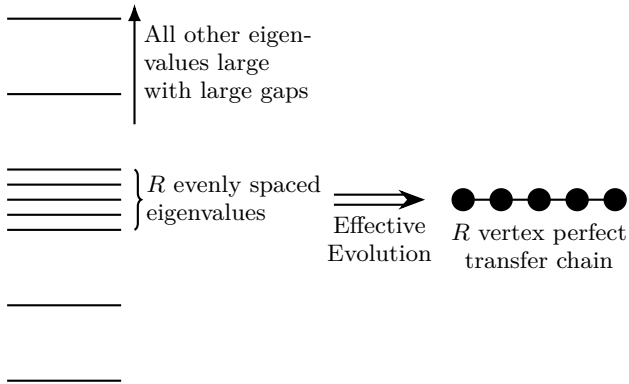
\begin{figure}
\centering
\begin{tikzpicture}[thick]
    \foreach \y in {0,1,-1,2,-2,7,12,-7,-12}
        \draw (0,0.2*\y) -- (1.5,0.2*\y);
    \draw[decoration={brace,mirror,raise=5pt},decorate]
  (1.5,-0.4) -- node[right=6pt,text width=2.5cm] {$R$ evenly spaced eigenvalues} (1.5,0.4);
  \draw [-Latex,xshift=5pt] (1.5,1) -- (1.5,2.6) node [midway,anchor=west,text width=2.5cm] {All other eigenvalues large with large gaps};
\foreach \x in {0,1,2,3,4}
    \node [circle,fill=black,minimum width=0.3cm] at (6+0.5*\x,0) {};
    \draw (6,0) -- (8,0) node [midway,anchor=north,text width=2.4cm,align=center,yshift=-0.2cm] {$R$ vertex perfect transfer chain};
\draw [-Stealth,double distance=3pt] (4.3,0) -- (5.5,0) node [midway,align=center,anchor=north,yshift=-3pt] {Effective\\Evolution};
\end{tikzpicture}
    \caption{For a chain of length $N$, we engineer perfect transfer by choosing a spectrum with odd integer gaps (left). To achieve a broad arrival width, we choose a central core of $R$ eigenvalues, and clear out the rest of them, making them large, and with large gaps. The end-to-end evolution is described by an effective length $R$ perfect transfer chain.
    }\label{fig:scheme}
\end{figure}

We will now design spin chains of arbitrary length $N$ with very broad arrival/departure characteristics which are sure to perform well for any reasonable function $p(t)$. The cost, as expected, is in terms of the state transfer time. We refer to our method, depicted in \cref{fig:scheme}, as the T-Rex as the net result is a chain with puny little arms at either end. We select a set of $R$ eigenvalues ($R$ and $N$ should have the same parity), with uniform spacing, centred on 0, gap 1 (e.g.\ $0,\pm 1,\pm 2,\ldots)$. We will then select the remaining $N-R$ eigenvalues, symmetrically about 0 and satisfying perfect state transfer conditions, but at much higher values, $\order{\gamma}$, each separated by $\order{\gamma}$ (e.g.\ $\pm\frac{\gamma}{2}$, $\pm\frac{3\gamma}{2}$, $\pm\frac{5\gamma}{2},\ldots$, where $\gamma\gg\frac{R}{2}$ and $\gamma$ is an odd integer). Having selected the target eigenvalues, an inverse eigenvalue problem for a symmetric tridiagonal system is easily solved \cite{kay2010a,karbach2005,gladwell2005}. The spectral structure is similar to that of \cite{bailey2025}, except that we actively retain a small number of eigenvalues close to 0.

The idea is that by \cref{eq:evecels} and normalisation, $\sum a_n=1$, the first $R$ eigenvectors have $a_n$ of order $1$, while the others, which we have cleared well away from the central region, are of $\order{\gamma^{1-R}}$ . At large $\gamma$, the effect of those large eigenvalues is negligible for an initial state $\ket{1}=\sum_n\sqrt{a_n}\ket{\lambda_n}$. The arrival statistics will be essentially be the same as those for a Krawtchouk chain of length $R$,
$$F_e\approx\sin^{2(R-1)}\left(\frac{\pi t}{2t_0}\right),$$
improving over the Krawtchouk's arrival profile, \cref{eq:Kraw_profile}.
\begin{figure}
    \centering
    \includegraphics[width=8cm]{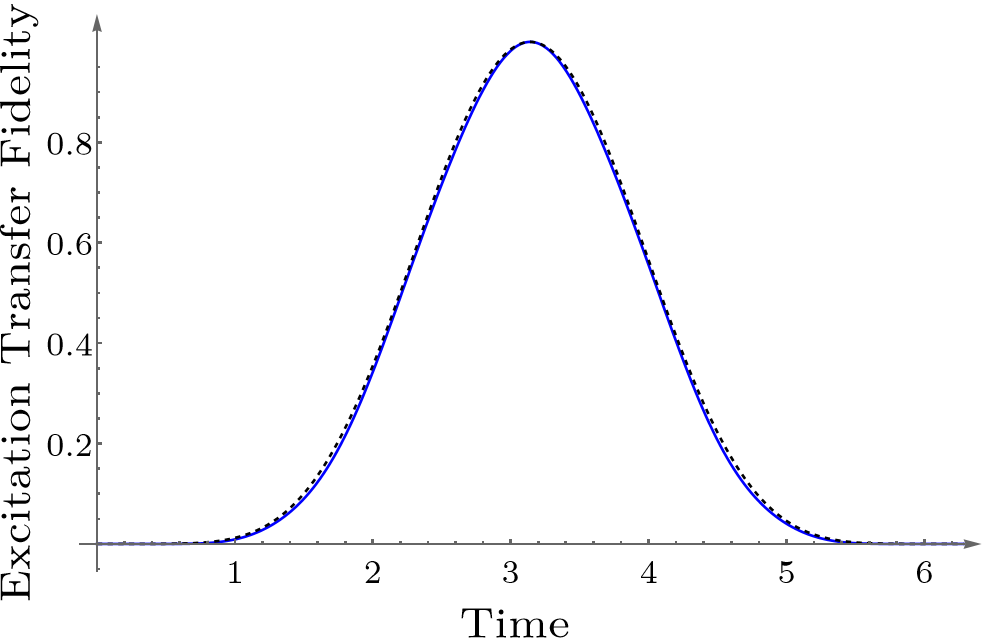}
\caption{Comparison of the excitation transfer fidelity for the T-Rex chain of length 8 (solid) with target arrival function (dashed). We chose to plot $\gamma=13$ since, for larger values, the two curves were essentially indistinguishable.}\label{fig:example}
\end{figure}
\begin{example}
Let us select a spectrum
$$
\pm\frac12,\pm\frac32,\pm\frac{\gamma}{2},\pm\frac{3\gamma}{2},
$$
aiming to create a length 8 chain with 
$$
F_e\approx\sin^{6}\frac{t}{2}.
$$
Solving the inverse eigenvalue problem yields a chain with the following couplings ($\gamma=149$)

\drawchain[0.98\textwidth]{0.8729, 10.54, 128.0, 150, 128.0, 10.54, 0.8729}

In \cref{fig:example}, we compare the evolution of the constructed T-Rex chain with the claimed $\sin^6$ functionality. The relevant code may be found in \cite{kay2025d}.

As a first approximation, we can describe the couplings as deriving from three different elements.
\begin{itemize}
\item The central $N-R$ spins are coupled by a $J=\gamma$ Krawtchouk chain $H_c$, accounting for the $\pm\frac{\gamma}{2},\pm\frac{3\gamma}{2}$ eigenvalues.
\item The two extremal couplings are the ends of a length $R=4$ Krawtchouk chain.
\item The remaining two couplings, $K$, are
$$
K^2\bra{1}H_c^{-1}\ket{N-R}=\frac{R}4.
$$
The right-hand side is the central coupling strength of the length $R$ Krawtchouk chain.
\end{itemize}
These approximations predict the chain to be

\drawchain[0.98\textwidth]{0.8660, 10.57, 129.0, 149, 129.0, 10.57, 0.8660}

\end{example}

By construction, the T-Rex chain is symmetric and has eigenvalues with odd integer gaps. It has perfect transfer at time $t_0=\pi$. We just need to assess how broad its arrival/departure peak is. We start by estimating the coupling strengths.

Starting at one end of the chain, we can evaluate
$$
\bra{1}H_0^k\ket{1}=\bra{N}H_0^k\ket{N}=\sum_na_n\lambda_n^k.
$$
Provided $k<R-1$, the $a_n\lambda_n^k$ values of the large eigenvalues vanish for large $\gamma$. The calculation, to a good approximation, is just that that one would perform for a PST chain of length $R$. The first $\lfloor\frac{R-2}{2}\rfloor$ coupling strengths are $\sqrt{n(R-n)}/2+\order{\gamma^{\frac{1+2n-R}{2}}}$ (\cref{eq:kraw_couple} with $J=\frac12$, determined  by the choice of gap size). These extremal arms are a puny $\order{1}$ strength, compared to the central region, which we will now see are a strength $\order{\gamma}$.

To evaluate the central couplings, calculate $\Tr(H_0^kS)=\sum_n(-1)^{n+1}\lambda_n^k$ progressively for increasing $k$ (some modification is required if $N$ is odd \cite{kay2016b}).
These must be entirely dominated by the large $\gamma$ terms for the length of an $N-R$ qubit chain (where the next coupling would be reported to be 0, and hence the small terms have a chance to be impactful).

\begin{figure}[!t]
\begin{tikzpicture}
    \node[anchor=south west,inner sep=0] (image) at (0,0) {\includegraphics[width=0.48\textwidth]{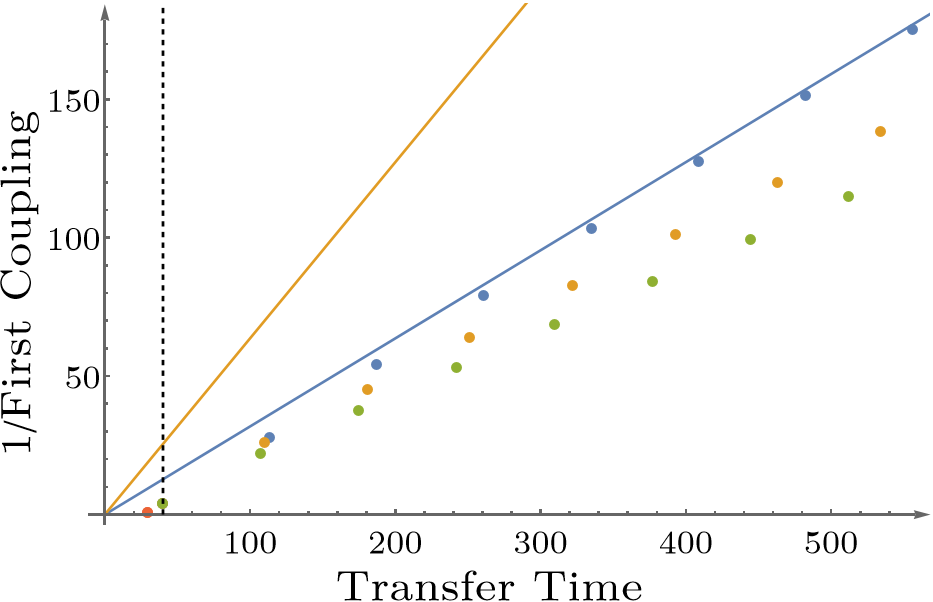}};
    \node [anchor=west,text width=2cm] at (1.5,5) {Optimal PST Time};
    \node [rotate=32] at (6,4.1) {Improved bound for odd $N\geq 5$};
    \node [rotate=51] at (3.5,4.2) {Mandelstam-Tamm};
    \draw [-Stealth] (5.5,1.5) node [anchor=west] {Apollaro} -- (1.42,0.83);
\end{tikzpicture}
\caption{Comparison of Mandelstam-Tamm limit, and improvement of \cref{thm:main} to that achieved by embedding $R=5,7,9$ (blue, orange, green) cases into a length 51 chain, with varying values of $\gamma$. The $R=5$ case, in the large time limit, tends to the bound. See \cite{kay2025d}.}\label{fig:tradeoff}
\end{figure}
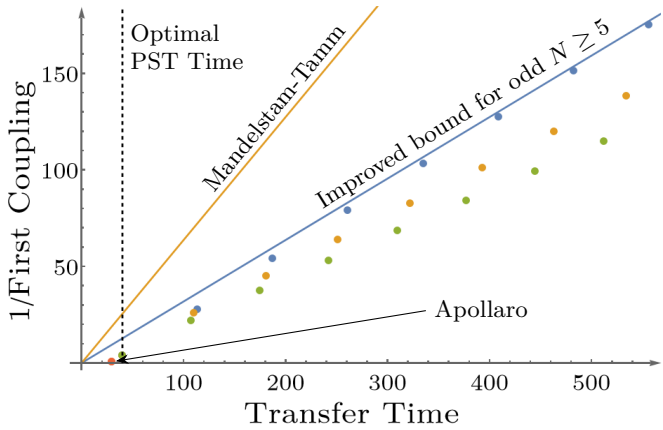

We must rescale everything by the maximum coupling strength, and track the change to the transfer time, $t_0\rightarrow \pi J_{\max}$. We claim that the largest coupling strength is the central one; in the limit of large $\gamma$, the regime of our primary interest, this is certainly true as the couplings of the Krawtchouk chain increase monotonically towards the centre, but whether it changes at finite $\gamma$ remains unclear. We can be exact about the central coupling ($N$ even), following \cite{yung2006}. Specifically,
\begin{align*}
2J_{N/2}&=\Tr(H_0S)=\sum_n\lambda_n(-1)^{n+1}\\&=\gamma\frac{N-R}{2}+(-1)^{\frac{N-R}{2}}\frac{R}{2}.
\end{align*}
This deviates from the crudely estimated approximation by only an amount $\frac{R}{2}$, which is small compared to $\gamma$.
After rescaling, we therefore achieve a value of
$$
J_1^2=\bra{1}H_0^2\ket{1}\approx\left(\frac{2\sqrt{R-1}}{\gamma (N-R)}\right)^2
$$
($R\geq 4$) with a perfect transfer time of
$$
t_0=\frac{\pi\gamma(N-R)}{4}.
$$
Thus, overall, we find
$$
J_1t_0\rightarrow\frac{\pi\sqrt{R-1}}{2}
$$
in the large $\gamma$ limit. For $R=4$, this is asymptotically optimal according to \cref{thm:main}. An equivalent calculation using $\Tr(H_0^2S)$ \cite{kay2016b} gives the same result for odd $N$, so that $R=5$ is asymptotically optimal. We can see this realised in \cref{fig:tradeoff}, where the $R=5$ case tends towards the limiting behaviour.

\subsection{Smaller \texorpdfstring{$R$}{R}}

Clearly, it is preferable to reach as small a value of $R$ as possible. We have easily argued the behaviour for $R\geq 4$, but wish to know more about the cases $R=2,3$.

We can still numerically evaluate the coupling coefficients of our T-Rex system in the cases of $R=2,3$, we just don't have the analytic backing. In these instances, indeed, we get evolution corresponding to $\sin^2$ and $\sin^4$ which, in their respective size cases (even, odd length chains respectively) are clearly optimal in terms of the broad arrival width. However, due to the difficult interactions of the small/large terms, they do not saturate the $\bra{1}H_0^2\ket{1}$ approximation. Presumably there are enough high-frequency variations to adversely affect the interpretation of $\bra{1}H_0^2\ket{1}$ even though, at longer timescales, we do not see these effects.

\begin{example}
Consider a symmetric chain of length 8 and spectrum
$$
\pm 1,\pm(1+2\gamma),\pm(1+4\gamma),\pm(1+6\gamma).
$$
For odd integer $\gamma$ this has perfect transfer in a time $t_0=\frac{\pi}{2}$. Once rescaled so that the maximum coupling strength is 1, this yields a T-Rex chain ($\gamma=51$):

\drawchain[0.98\textwidth]{0.086, 0.866, 0.712, 1, 0.712, 0.866, 0.086}

In \cref{fig:P2}, we plot this case, chosen at an intermediate regime to emphasise why the $\bra{1}H_0^2\ket{1}$ term fails. In the large $\gamma$ limit, the $\sin^2$ behaviour emerges, but the transfer time is impacted. While the Mandelstam-Tamm limit conveys that a transfer with this width could be implemented in a time as short as $t_0=18.3$, this takes a factor $\order{\gamma}$ longer at $t_0\approx320$.
\end{example}

\begin{figure}[t]
    \centering
    \begin{tikzpicture}
     \node[anchor=south west,inner sep=0] (image) at (0,0) {\includegraphics[width=0.43\textwidth]{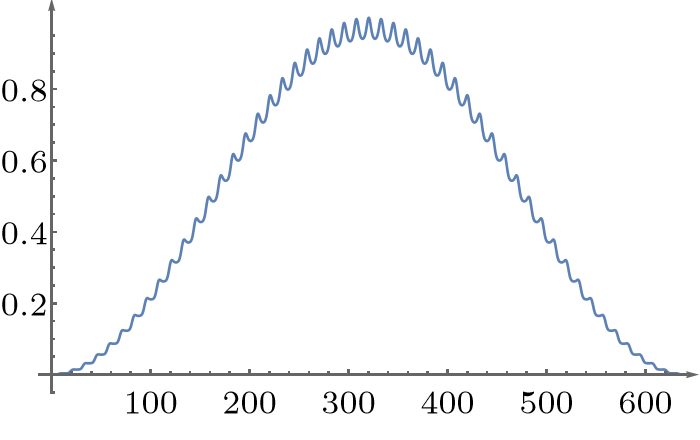}};
     \node [below=-0.2cm of image] {\large Time};
     \node [left=0.15cm of image,rotate=90,anchor=south] {\large Excitation Transfer Fidelity};
    \end{tikzpicture}
    \caption{A T-Rex system where $R=2$ and $N=8$. Although $J_1^2=\bra{1}H_0^2\ket{1}$ is comparatively large, implying poor arrival, this is due to the thin peaks. As $\gamma\rightarrow\infty$, these peaks vanish and the overall optimal $\sin^2$ dependence emerges.}\label{fig:P2}
\end{figure}

While the case of $R=2$ suffers from a $\gamma$ blow-up of the transfer time in achieving its $\sin^2$ behaviour, the $R=3$ case achieves the near-optimal $\sin^4$ behaviour with only an $\order{1}$ multiplier to the transfer time (observed numerically). To see why, imagine that the extremal couplings $J_1=J_{N-1}\sim\frac{1}{\gamma}$ are small, and are the only small ones. In the $R=2$ case, an analysis at second order of degenerate perturbation theory (because the central chain does not have a 0 eigenvalue) conveys that there are two eigenvalues at approximately $\pm J_1^2$, and so perfect transfer must take a time at least $\order{J_1^{-2}}$. On the other hand, when $R=3$, the central chain is of odd length, and has a 0 eigenvalue. The degenerate perturbation theory analysis proceeds at first order; there is a splitting $\order{J_1}$ and so the evolution takes a time $\order{J_1^{-1}}$.

\section{Robustness to Perturbations}

\begin{figure}
    \centering
    \includegraphics[width=0.48\textwidth]{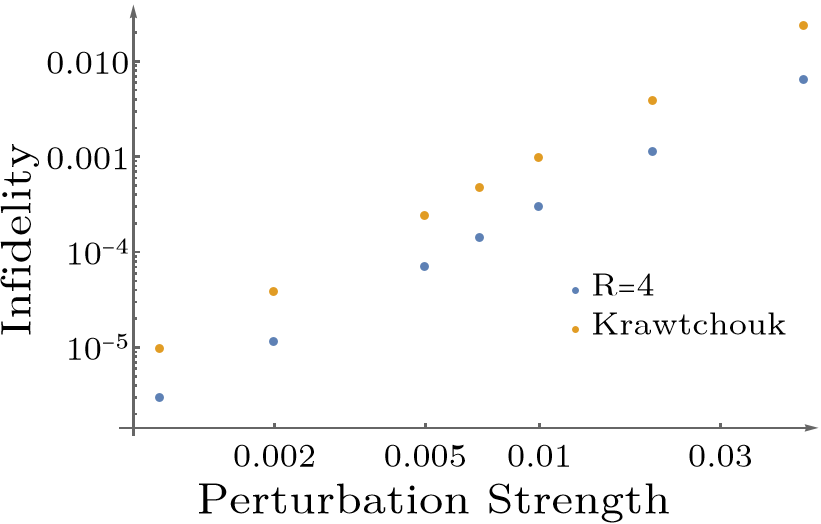}
    \caption{Comparison of robustness to perturbations between Krawtchouk chain and an $R=4$ T-Rex ($\gamma=21$) for chains of length 50 with maximum coupling strength of 1. Perturbations are only applied to the central 45 couplings. Plotted values are the upper quartile of $1-\sqrt{F_e}$ after 10000 samples (independently chosen uniformly at random in the range $\pm\delta$, for perturbation strength $\delta$), in order to represent the process of manufacturing multiple samples, testing them, and rejecting the worst. Smaller is better.}\label{fig:perturb}
\end{figure}

The focus of this work has been to has been to create a transfer system with a broad arrival peak. We have done this effectively by fusing two chains together; a short section (e.g.\ $R=4,5$) which is divided into two halves, becoming the puny arms of the T-Rex, attached at either end of a second chain, of much greater coupling strengths, mediating communication. Effectively, perfect state transfer is achieved on the short chain, almost entirely ignoring the second chain except for this mediation it creates. So, what happens if, when we manufacture this chain, the central section is a little bit wrong? \cref{fig:perturb} conveys that, numerically, the effect on the transfer would be minimal, yielding almost an order of magnitude improvement! This is in spite of the fact that the $R=4$ chain has over 20 times longer than the Krawtchouk chain to accumulate the adverse effects of the perturbation.

\section{Encoded Transfer}

Another scenario in which the quality of transfer can be optimised is the use of encoding \cite{haselgrove2005} (see also \cite{keele2021,keele2021a,kay2022,kay2022a,kay2009}). Here, we identify a set of vertices $A$ on which a single-excitation state can be encoded, and a set of sites $B$ on which it is received. The optimal states for input and output at time $t$ are the right- and left-singular vectors of the maximum singular value of the operator
$$
\Pi_Be^{-iH_0t}\Pi_A
$$
where $\Pi_X=\sum_{i\in X}\proj{i}$.
This is easily modified to take into account timing inaccuracy since
$$
\tilde F_e=\int p(t)\bra{\psi_\text{dec}}\Pi_Be^{-iH_0t}\Pi_A\ket{\psi_{\text{enc}}}dt.
$$
$\tilde F_e$ is optimised by selecting the singular vectors corresponding to the maximum singular value of the operator
$$
M=\int p(t)\Pi_Be^{-iH_0t}\Pi_Adt.
$$
Computationally, it is easier to work with the approximation
$
\tilde M=-\Pi_BH_0^2e^{-iH_0t}\Pi_A,
$ 
but this presupposes that the encoding yields perfect transfer at the time $t_0$.

When applied to the uniform chain \cite{osborne2004,kay2009}, this optimisation makes very little difference to the known analytic result for maximising the fidelity of transfer because the transferring wavepacket already has a very broad peak, and is only supported on the low-frequency eigenvectors (the linear part of the spectrum). Nevertheless, we can see that improvements are possible in other cases.

\subsection{Case Study}

\begin{figure}
\centering
\begin{tikzpicture}
    \node[anchor=south west,inner sep=0] (image) at (0,0) {\includegraphics[width=0.48\textwidth]{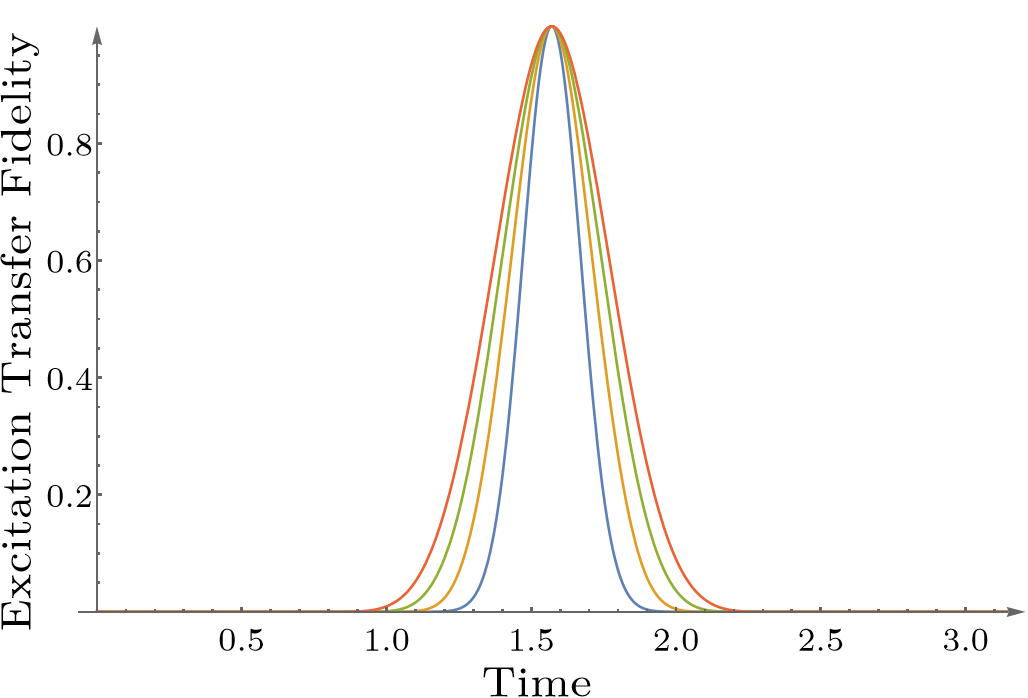}};
    \draw [-Stealth,thick] (4.7,3) -- (5.5,3.4) node [anchor=west,text width=3cm] {Increasing size of encoding region};
\end{tikzpicture}
    \caption{Arrival of a state through a length 51 perfect transfer chain with encoding/decoding regions of size ($1,3,5,7$).}\label{fig:enc}
\end{figure}

We shall take as a case study a perfect transfer chain of length $N$, and encoding and decoding regions of size $M$ (odd) at opposite ends of the chain. Any initial state restricted to the encoding region $A$ must transfer perfectly to the decoding region $B$ in the perfect transfer time $t_0$. We are therefore free to pick the encoding to optimise the timing insensitivity. The optimal may be found by finding the smallest singular value of
$$
-\tilde M=\Pi_ASH_0^2\Pi_A.
$$
We numerically plot some results in \cref{fig:enc}, observing how the arrival peak steadily broadens as we increase the size of the decoding region.
An alternative approach, (which must be sub-optimal \footnote{The $\{a_n\}$ for the eigenvectors that we remove are small, making a negligible contribution to the weighted average $\sum a_n\lambda_n^2$, so removing them completely will be less effective than removing weight from some other terms.}, but we find interesting for the connections it makes), is to consider the eigenvectors $\ket{\lambda_n}$ of $H_0$. Similar to the approach for achieving perfect transfer by encoding \cite{kay2022a,kay2022}, we find the state supported on the input site that is orthogonal to $\{\ket{\lambda_n},\ket{-\lambda_n}\}_{n=1}^{(M-1)/2}$. Since this encoding is no longer supported on the highest energy eigenvectors, the highest frequency components are eliminated, and the arrival/departure peak must be broadened. This strategy is effectively achieving the same as \cref{lem:family}, except without physically changing the chain. Were we to apply it to the uniform chain, it would effectively be achieving the results of \cite{vinet2024}. The advantages seem negligible, however -- by increasing the size of the encoding region, you are decreasing the transfer distance and you are only improving the arrival/departure width to the same extent that shortening the chain to the transfer distance would achieve. On the other hand, the advantage of encoding is that it can be applied after manufacture, and can adapt to manufacturing imperfections.

\section{Time Insensitive Fractional Revival}

In \cite{kay2010a,dai2010}, a conversion was given such that any odd length chain possessing perfect excitation transfer $\ket{1}\rightarrow\ket{N}$ can be modified into one with fractional revival,
$$
\ket{1}\rightarrow\cos2\theta\ket{1}+\sin\theta\ket{N}.
$$
All that needs to be changed is the central two coupling strengths are updated
\begin{align*}
J_{(N-1)/2}&\rightarrow J_{(N-1)/2}\sqrt2\cos(\theta),\\ J_{(N+1)/2}&\rightarrow J_{(N+1)/2}\sqrt2\sin(\theta).
\end{align*}
The rest are unchanged, as is the evolution time $t_0$. Using the T-Rex construction of \cref{sec:clearout} must realise a solution with an asymptotically optimal arrival profile.

\begin{figure}[t]
\centering
\includegraphics[width=8cm]{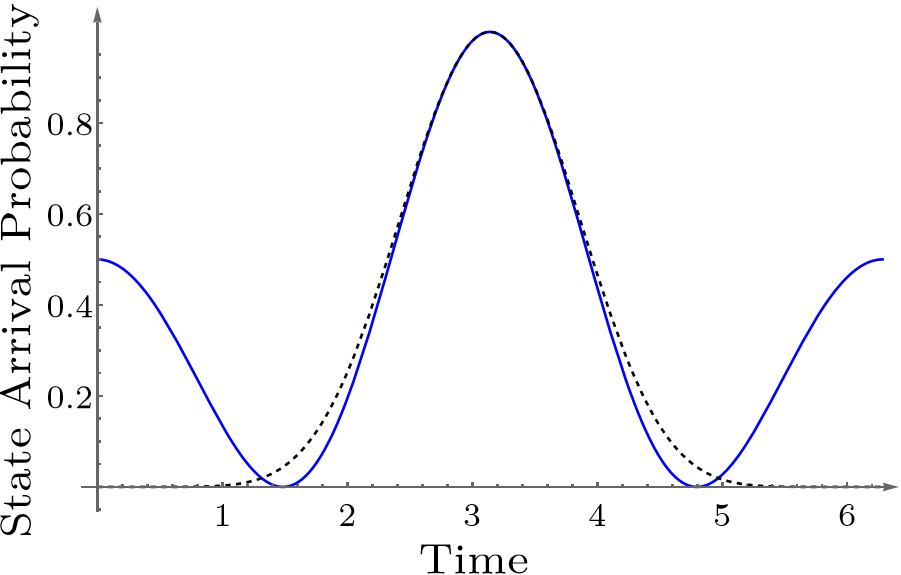}
    \caption{Fractional revival on a chain of length 11, using a T-Rex chain (blue) and the ideal $\sin^8$ behaviour (dashed).}\label{fig:fr}
\end{figure}

\begin{example}
Consider a T-Rex perfect transfer chain with $N=9$, $R=5$, $t_0=\frac{\pi}{2}$ and $\gamma=11$. We only change the central two couplings $J_5=J_6$ into $\tilde J_5=\sqrt{2}J_5\cos\frac{\pi}{8}$ and $\tilde J_6=\sqrt{2}J_5\sin\frac{\pi}{8}$.

\drawchain[0.98\textwidth]{1.01, 2.04, 12.8, 18.7, 7.75, 12.8, 2.04, 1.01}

In time $t_0$, we now create the evolution
$$
\ket{1}\rightarrow \frac{1}{\sqrt{2}}(\ket{1}+\ket{9}).
$$
If we measure the probability of this state having arrived as a function of time, then close to the arrival time, the probability behaves like the optimal $\sin^8$ behaviour (away from this, there is some deviation because the initial state is not orthogonal to the target state). This evolution is depicted in \cref{fig:fr}.    
\end{example}

There is a second method for creating fractional revival chains \cite{genest2016}. In this case, we take the spectrum of a perfect state transfer chain. The eigenvalues are alternately associated with the symmetric and anti-symmetric subspaces. In the perfect transfer time, both exhibit a perfect revival, with a relative phase of $-1$. If we shift the spectrum of the antisymmetric subspace relative to the symmetric one, both still exhibit perfect revivals, but with a different relative phase $\theta$. This means that the state at this time is
$\frac12((\ket{1}+\ket{N})+e^{i\theta}(\ket{1}-\ket{N}))$
which is, again, a perfect revival. The T-Rex chains can be adapted in the same way, providing fractional revival chains with broad arrival/departure peaks that are asymptotically optimal for their insensitivity to timing errors. This construction is not limited by the parity of the chain length.

\section{Summary}

In this paper, we have resolved how insensitive a perfect state transfer chain can be made to timing errors. Unlike the original, na\"ive, analysis, we achieve a broad arrival peak of the form $F_e=\sin^6(t)$ or $\sin^8(t)$ for any even or odd length of chain respectively. We have proven that this is asymptotically optimal by specialising the Mandelstam-Tamm bound to the specific circumstance of perfect state transfer. The resulting chains are also very robust against perturbations. We also applied the same T-Rex construction to fractional revivals. Examples of all of these may be found in \cite{kay2025d}.

The limit at which we work moves the operating regime of these spins chains away from the useful working limit of the optimal perfect state transfer time. The optimal performance at these shorter times remains open, with a modest gap between the current best-known chains and what the limit of \cref{thm:main} allows. Moreover, \cref{thm:main} does not cover the case of high fidelity transfer chains which to not achieve perfect transfer, but may achieve their transfer on shorter timescales, such as the Apollaro chain \cite{apollaro2012}, or perhaps the chains that interpolate between uniform and PST \cite{vinet2024}. While the Apollaro chain has a shorter arrival time, its arrival peak is thinner, so the net performance is worse, as indicated in \cref{fig:tradeoff}. Perhaps these chains could be further designed to improve performance over the perfect transfer ones, whether by achieving a similar arrival profile at a shorter time, or better arrival profile at the same time. There is some potential for a reanalysis of \cite{apollaro2012}: if we define a sub-optimal target fidelity, then one can vary over the available parameters that all achieve this baseline fidelity, and find the one with the broadest arrival peak. We leave this, or other routes towards improved arrival profile for high fidelity transfer, for future investigation.

The insights of the T-Rex construction are profound, and extend well beyond the present aim. In state transfer, we can also apply the ideas to uniform networks, while in algorithmic scenarios, the puny arms are remarkably effective in both quantum search and matrix inversion. These will be elucidated in future works.

\bigskip
\section*{Acknowledgments}

C.T.\ is supported by National Science Foundation grant OSI-2427020.
We thank Ada Chan and Pierre-Antoine Bernard for helpful discussions.

%
\end{document}